\documentclass[12pt]{iopart}

\expandafter\let\csname equation*\endcsname\relax
\expandafter\let\csname endequation*\endcsname\relax
\usepackage{amssymb}
\usepackage{amsmath}
\usepackage[normalem]{ulem}
\usepackage{mathrsfs,amsbsy,color,xcolor,graphicx,bm,amsthm,amsfonts}%
\usepackage{wasysym,units,wrapfig}
\usepackage{algpseudocode,algorithm}
\usepackage{subfigure}
\usepackage{dsfont}
\begin{document}

\newtheorem{theorem}{Theorem}
\newtheorem{proposition}[theorem]{Proposition}
\newtheorem{corollary}[theorem]{Corollary}

\title[Quantum circuit design for accurate simulation of qudit channels]{Quantum circuit design for accurate \\
simulation of qudit channels}

\author{Dong-Sheng Wang$^1$ and Barry C. Sanders$^{1,2}$}

\address{$^1$Institute for Quantum Science and Technology, University of Calgary, Alberta T2N 1N4, Canada}
\address{$^2$Program in Quantum Information Science,
	Canadian Institute for Advanced Research, Toronto, Ontario M5G 1Z8, Canada}
\ead{dongshwa@ucalgary.ca}
\ead{sandersb@ucalgary.ca}

\vspace{10pt}

\begin{abstract}
    We construct a classical algorithm that designs quantum circuits
	for algorithmic quantum simulation of arbitrary qudit channels
    on fault-tolerant quantum computers
	within a pre-specified error tolerance with respect to diamond-norm distance.
	The classical algorithm is constructed by decomposing a quantum channel
	into a convex combination of generalized extreme channels by optimization of a set of
	nonlinear coupled algebra\"{i}c equations.
	The resultant circuit is a randomly chosen generalized extreme channel circuit
	whose run-time is logarithmic with respect to the error tolerance
	and quadratic with respect to Hilbert space dimension,
	which requires only a single ancillary qudit plus classical dits.
\end{abstract}

%
%
%
%

\section{Introduction}
\label{sec:intr}

Algorithmic quantum simulation~\cite{Fey82,Llo96,BACS07a,BN09},
which is digital quantum simulation with pre-specified bounded-error output~\cite{Sand13},
is important for simulating many-body dynamics~\cite{WBHS11,RWS12},
quantum-state generation and
dissipative quantum-state engineering~\cite{VWC09,MDPZ12},
quantum thermodynamics~\cite{TD00,PW09},
nonequilibrium quantum phase transitions~\cite{SAA+10,SMN+13},
testing element distinctness~\cite{Chi09},
and solving linear equations~\cite{HHL09}
and differential equations~\cite{Ber14}.
Experimental quantum simulation~\cite{BN09}
has been demonstrated in quantum computing implementations such as
ion traps~\cite{LHN+11,BSK+12,BMS+11},
atoms in optical lattice~\cite{LSA+07,WMBL11},
and superconducting circuits~\cite{GSP14}.
Whereas unitary evolution generated by a self-adjoint Hamiltonian
has so far been the major research focus,
algorithmic quantum simulation of quantum channels
(i.e., completely-positive trace-preserving mappings)~\cite{Sti55,Cho75,Kra83}
and open-system dynamics~\cite{NC00} is a nascent and exciting research area
both theoretically~\cite{KBG+11,WBOS13} and experimentally~\cite{SMN+13,FPK+12}.
Quantum channel simulators can play vital roles in quantum simulation for the study of
quantum non-Markovian effects~\cite{CGM+12,SPA+12},
dissipative quantum many-body dynamics~\cite{SMN+13},
and also for modelling quantum noise for the test of various protocols~\cite{TB05,MPG+13}.

Our aim is to develop a classical algorithm
that designs quantum circuits,
which simulate accurately an arbitrary qudit
(corresponding to states described by positive semidefinite trace-class operators~$\rho$
acting on the~$d$-dimensional Hilbert space $\mathscr{H}_d\cong\mathbb{C}^d$)
channel~$\mathcal{E}$
on fault-tolerant quantum computers
within a specified error tolerance~$\epsilon$
with respect to the diamond norm~\cite{AKN97}.
This quantum circuit needs to be executed efficiently with respect to~$\epsilon$,
i.e., with quantum space and time resources scaling as polylog$\frac{1}{\epsilon}$.
We treat classical space resources, i.e.~dits ($d$-dimensional digits, $d=2$ for bits), as free.
Previously two different cases have been considered:
Markovian qudit channel simulation
but with the Liouvillian rather than the channel as input~\cite{KBG+11}
and for a general single-qubit channel
as input with an efficient simulation circuit as output~\cite{WBOS13}
employing qubit extreme channel theory~\cite{Cho75,RSW02}.
Efficient and accurate algorithmic quantum simulation,
such that the output is delivered with minimal resources and within the pre-specified error tolerance,
is vital for constructing quantum simulators in the near term that answer computational
problems.

Generalizing from simulating qubit channels to qudit channels is not straightforward
because decomposing an arbitrary qudit channel into a convex combination of generalized extreme
channels is an open problem in quantum information~\cite{Rus07}.
We circumvent this obstacle by decomposing approximately, rather than exactly,
into a convex sum of generalized extreme channels,
and we construct a classical optimization algorithm~\cite{BV04}
that devises circuits for simulating generalized extreme channels so that the entire
qudit channel can be simulated by random concatenations of generalized extreme channel simulators.
The circuits devised by our algorithm show how to realize
algorithmic quantum-channel simulation.

We approach the problem of constructing a qudit channel circuit simulator
by constructing an algorithm whose inputs are the description of~$\mathcal E$,
the tolerance~$\epsilon$,
and the dimension~$d$ of the qudit Hilbert space
and delivering an output comprising a description of a simulation circuit
and the actual error~$\tilde{\epsilon}$.
In contrast to the case for unitary channels,
which can be constructed as a concatenation of other unitary channels,
the non-unitary channel is not such a simple sequence of channels~\cite{WC08}
thereby resulting in complicated approach to quantum channel simulation.

A direct procedure for a quantum simulation of a quantum channel is to
employ Stinespring dilation~\cite{Sti55},
which replaces the qudit channel $\mathcal{E}$
by a unitary channel $\mathcal{U}$,
with up to~$d^3$ dimensions,
followed by a partial trace over the environment
to recover the description of the channel $\mathcal{E}$.
The circuit for simulating the channel by
a dilated unitary channel acting on a Hilbert space of dimension~$d^3$
generically requires $O(d^6)$
single-qubit and two-qubit gates~\cite{BBC+95,MV06},
obtained from a small universal instruction set using a Solovay-Kitaev gate decomposition
approach~\cite{DN06,KMM13b,Kli13}.
Therefore, the time cost is $O(d^6)$ and the space cost is three qudits.

In the interest of bringing algorithmic quantum-channel simulation to its lowest possible
cost for experimental expediency,
we employ the procedure of approximately decomposing the channel
into a convex combination of generalized extreme channels.
The simulation circuit has a time cost of $O(d^2\log \frac{d^2}{\epsilon} )$,
which is the same as the time cost for simulating a unitary qudit channel,
hence a lower bound~\cite{BBC+95,BOB05}.
Furthermore this procedure requires a spatial resource of just two qudits
plus random dits,
hence reduces the quantum space cost by a third.

This work comprises three parts.
In Sec.~\ref{sec:extdecomp} we present our method for
the construction of extreme channels.
In Sec.~\ref{sec:circ} we construct quantum circuits for generalized extreme channels.
In Sec.~\ref{sec:algo} we discuss the quantum channel simulation algorithm.
We conclude briefly and provide supporting information in Appendix.

\section{Extreme quantum channels}
\label{sec:extdecomp}

A quantum channel $\mathcal{E}\in\mathscr{S}_d$,
the set of all channels for qudits of dimension~$d$,
can be represented as
\begin{equation}
\label{eq:channel}
	\mathcal{E}(\rho)
		=\sum_{i=0}^m K_i\rho K_i^\dagger
\end{equation}
for all states~$\rho$
with a set of linearly independent Kraus operators~\cite{Kra83}
\begin{equation}
	\{K_i:\mathscr{H}_d\to\mathscr{H}_d\}
\end{equation}
such that
\begin{equation}
	\sum_{i=0}^m K_i^\dagger K_i=\mathds{1},
\end{equation}
and Kraus rank~\cite{Kra83}
\begin{equation}
	m+1\leq d^2.
\end{equation}
A channel is extreme if and only if it cannot be written as a convex sum of other channels.
Equivalently a channel is extreme if and only if $m$ is bounded above by~$d-1$
and $\{K_i^\dagger K_j\}$ is a linearly independent set~\cite{Cho75}.

A channel is called a generalized extreme channel if its Kraus rank is at most~$d$,
and a generalized extreme channel which is not extreme is called a quasi-extreme channel~\cite{Rus07}.
Clearly the set of generalized extreme channels contains both extreme channels and quasi-extreme channels.

Next we propose a Kraus operator-sum representation
for an arbitrary rank-$d$ extreme or quasi-extreme channel.
First, we construct the sum representation
using the Heisenberg-Weyl basis
\begin{equation}
	\left\{X_i Z_j;i,j\in\mathbb{Z}_d\right\}
\end{equation}
for 	
\begin{equation}
	X_i =\sum_{\ell=0}^{d-1}|\ell\rangle\langle \ell+i|,\;
	Z_j =\sum_{\ell=0}^{d-1}\operatorname{e}^{\operatorname{i}2\pi \ell j/d}|\ell\rangle\langle \ell|.
\end{equation}
\begin{proposition}
\label{prop:ext}
A rank-$d$ extreme channel $\mathcal{E}^\text{e}\in\mathscr{S}_d$
can be represented by
\begin{equation}
\label{eq:extKraus}
	\mathcal{E}^\text{e}(\rho)=\sum_{i=0}^{d-1}K_i\rho K_i^\dagger
\end{equation}
for any Kraus operators satisfying
\begin{equation}
\label{eq:F}
	K_i:=W F_i V,\;
	F_i:=X_iE_i,\;
	E_i:=\sum_{j=0}^{d-1}a_{ij}Z_j,\,
	i\in\mathbb{Z}_d,
\end{equation}
for any unitary operators
\begin{equation}
	V, W \in SU(d),
\end{equation}
provided that
\begin{equation}
	\{a_{ij}\in \mathbb{C}\}
\end{equation}
is chosen such that the set $\{F_i^\dagger F_j\}$
is linearly independent and
\begin{equation}
	\sum_{i=0}^{d-1} F_i^\dagger F_i=\mathds{1}
\end{equation}
is satisfied.
\end{proposition}
\begin{proof}
Per definition, Eq.~(\ref{eq:extKraus}) holds for any rank-$d$ extreme channel
with~$\{K_i^\dagger K_j\}$ being linearly independent.
Thus, the proof focuses on showing that the ansatz~(\ref{eq:F})
yields arbitrary linearly independent operators $\{K_i^\dagger K_j\}$.

Linear independence of $\{K_i^\dagger K_j\}$
requires that
\begin{equation}
\label{eq:Xi}
	\Xi:=\sum_{i,j=0}^{d-1}\gamma_{ij}K_i^\dagger K_j=0
	\iff\gamma_{ij}\equiv 0 \; \forall i,j.
\end{equation}
From Eq.~(\ref{eq:F}),
\begin{equation}
	\Xi=V^\dagger \left(\sum_{ij}\gamma_{ij}F_i^\dagger F_j\right)V.
\end{equation}
This is a unitary conjugation of the sum in parentheses
so we ignore $V$ in the proof.
Therefore, we need to require linear independence of $\{F_i^\dagger F_j\}$. For
\begin{equation}
\label{eq:b}
	b_{i\mu\nu}:=\sum_{k,l=0}^{d-1} a^*_{ik}a_{i+\mu,l} \operatorname{e}^{\operatorname{i}2\pi[\mu l+\nu(l-k)]/d},\;
	\mu\in\mathbb{Z}_d,
\end{equation}
we observe that
\begin{equation}
	F_i^\dagger F_{i+\mu}=\sum_{\nu=0}^{d-1} b_{i\mu\nu} |\nu\rangle\langle \nu+\mu|;
\end{equation}
hence
\begin{equation}
	\text{tr}[(F_i^\dagger F_{i+\mu})^\dagger F_j^\dagger F_{j+\mu'} ]=0
\end{equation}
for $\mu\neq \mu'$.
Now we partition
\begin{equation}
\label{eq:partition}
	\left\{F_i^\dagger F_{i+\mu};i,\mu\in\mathbb{Z}_d\right\}
		\to\left\{\{F_i^\dagger F_{i+\mu};i\in\mathbb{Z}_d\};\mu\in\mathbb{Z}_d\right\}.
\end{equation}
For $\{F_i^\dagger F_{i+\mu}\}$ to be a linearly independent set,
each subset must be linearly independent.
For each subset,
\begin{equation}
	\Xi_\mu:=\sum_{i=0}^{d-1} \gamma_{i,i+\mu} F_i^\dagger F_{i+\mu}
\end{equation}
so
\begin{equation}
	\Xi =\sum_{\mu=0}^{d-1} \Xi_\mu.
\end{equation}
Then
$\Xi\equiv 0$
implies $\Xi_\mu \equiv 0 \; \forall \mu$.

Now we establish linear independence of $\{F_i^\dagger F_j\}$
by constraining each subset~(\ref{eq:partition}).
First we map each matrix $F_i^\dagger F_{i+\mu}$
to a vector
\begin{equation}
	\bm{b}_{i\mu}:=(b_{i\mu\nu}).
\end{equation}
Then linear independence of
\begin{equation}
	\{F_i^\dagger F_{i+\mu}\}
\end{equation}
is ensured by the condition that the determinant of each matrix
\begin{equation}
	B_\mu:=(\bm{b}_{i\mu})
\end{equation}
is nonzero;
i.e.,
\begin{equation}
	\text{det} B_\mu\neq 0 \; \forall \mu
\end{equation}
(except for a zero-measure subset of $\{a_{ij}\}$).
Then $\Xi_\mu\equiv 0$ implies
\begin{equation}
	\gamma_{i,i+\mu}\equiv 0 \; \forall i,\mu,
\end{equation}
which establishes linear independence of $\{F_i^\dagger F_j\}$
hence also $\{K_i^\dagger K_j\}$.

As $\{F_i^\dagger F_j\}$ spans $\mathcal{B}(\mathscr{H}_d)$,
the space of bounded linear operators on $\mathscr{H}_d$,
hence is a basis.
Composition with~$V$ and also $W$ ensures that an arbitrary basis
$\{K_i^\dagger K_j\}$ can be realized for an extreme channel.
Consequently,
the proof showing extremality of the channel~(\ref{eq:extKraus}) is complete.
\end{proof}

\begin{corollary}
The set of Kraus operators~$F_i\left(\{a_{ij}\in\mathbb{C}\}\right)$
has at most $d^2-d$ independent real parameters.
\label{cor:parameters}
\end{corollary}
\begin{proof}
We prove the statement using the property of the Choi-Jamio{\l}kowski state~\cite{Cho75,Jam72}
\begin{equation}
\label{eq:C}
	\mathcal{C} := \mathcal{E} \otimes\mathds{1} (|\eta\rangle\langle\eta|),\;
	|\eta\rangle=\sum_{i=0}^{d-1}|i\rangle|i\rangle \in \mathscr{H}_d \otimes \mathscr{H}_d,
\end{equation}
for a channel $\mathcal{E}$ and~$\{|i\rangle\}$ the computational basis of $\mathscr{H}_d$.
With Prop.~\ref{prop:ext}
and defining $\tilde{a}_{i,l+i}:=\sum_{j=0}^{d-1}a_{ij}\operatorname{e}^{\operatorname{i}2\pi(l+i) j/d}$,
we find that
the Choi-Jamio{\l}kowski state $\mathcal{C}^\text{e}$ corresponding to $\{F_i\}$ is
\begin{equation}\label{eq:extchoi}
  \mathcal{C}^\text{e}=\sum_{i=0}^{d-1} \sum_{k,l=0}^{d-1}\tilde{a}^*_{i,k+i} \tilde{a}_{i,l+i}|l,l+i\rangle\langle k,k+i|,
\end{equation}
which is a~$d$-sparse, rank-$d$ positive semidefinite matrix
with at most $d^2$ real parameters.
Constrained by normalization,
$\{F_i\}$ has at most $d^2-d$ independent parameters.
\end{proof}

When the set $\{K_i^\dagger K_j\}$ is not linearly independent
our construction~(\ref{eq:F}) yields non-extreme yet quasi-extreme channels,
which are in the closure of the set of extreme channels~\cite{Rus07}.
As mentioned in the proof of Prop.~\ref{prop:ext},
the set of extreme channels dominates the set of all generalized extreme channels.
Also, both extreme and quasi-extreme channels with rank smaller than~$d$ can be realized if
some of the Kraus operators are zero matrices.
\begin{corollary}
\label{coro:quasiext}
A rank-$d$ generalized extreme channel $\mathcal{E}^\emph{g}\in\mathscr{S}_d$
can be represented by
\begin{equation}
	\mathcal{E}^\emph{g}(\rho)=\sum_{i=0}^{d-1}K_i\rho K_i^\dagger
\end{equation}
with Kraus operators~(\ref{eq:F})
for any unitary operators $V, W \in SU(d)$
and
\begin{equation}
	\sum_{i=0}^{d-1} F_i^\dagger F_i=\mathds{1}.
\end{equation}
\end{corollary}
\begin{proof}
  From the construction~(\ref{eq:F}), it holds $\text{tr} F_i^\dagger F_j =0$ for $i\neq j$,
  which means the set $\{F_i\}$ (also $\{K_i\}$) is linearly independent.
  The unitary operators $V$ and $W$ can take this set to an arbitrary linearly independent set
  with the same cardinality.
  This proves that any rank-$d$ channel can be written in the proposed form.
\end{proof}

\section{Quantum circuits for generalized extreme channels}
\label{sec:circ}

Now that we have the sum representation of the (quasi)extreme channel~$\mathcal{E}^\text{e}$
with respect to $\{F_i\}$~(\ref{eq:F}),
we construct a quantum circuit for simulating state evolution through the channel
by employing Stinespring dilation~\cite{Sti55}.
The Kraus operators $\{F_i\}$
can be realized by a channel
\begin{equation}
\label{eq:Uchannel}
	\mathcal{U}(\bullet):= U\bullet U^\dagger,\;
	U\in SU(d^2)
\end{equation}
with~$U$ acting on the system~(s) qudit and
an ancillary~$d$-dimensional ancilla~(a) qudit
such that $F_i=_\text{a}\!\!\!\langle i| U|0\rangle_\text{a}$.
Such Kraus operators~$\{F_i\}$ trivially satisfy
linear independence and
\begin{equation}
	\sum_iF_i^\dagger F_i=\mathds{1}.
\end{equation}

In principle the quantum circuit for a generalized extreme channel could be constructed in three stages:
solve the Kraus decomposition in Prop.~\ref{prop:ext},
then use the Kraus operators~$\{F_i\}$ to construct the unitary operator
based on Stinespring dilation,
and finally decompose it into a quantum circuit comprising gates from a finite
universal instruction set.
However, this method is stymied by the intractability of the nonlinear algebra\"{i}c equations that arise
from the Kraus decomposition so this approach is not viable.

Instead we adopt a different tack,
which is to find the quantum circuit by optimization.
In this approach, we construct the circuit
for a generalized extreme channel as a sequence of instruction-set gates
and optimize over the set of circuits
such that
the diamond-norm distance~\cite{AKN97}
(rather than the induced Schatten one-norm~\cite{KBG+11,WBOS13})
\begin{equation}
	\|\mathcal{E}-\mathcal{\tilde{E}}\|_{1\rightarrow 1}
		:=\max_\rho\|\mathcal{E}(\rho)-\mathcal{\tilde{E}}(\rho)\|_1
\end{equation}
between the input channel~$\mathcal E$ and
the approximate channel~$\tilde{\mathcal E}$ satisfies
\begin{align}
\label{eq:diamond}
	\|\mathcal{E}-\mathcal{\tilde{E}}\|_\diamond
		:=\|\mathcal{E}\otimes\mathds{1}-\mathcal{\tilde{E}}
			\otimes\mathds{1}\|_{1\rightarrow 1}\leq \epsilon.
\end{align}
The diamond-norm distance is preferred
as it gives worst-case gate error,
and has the operational meaning that the probability of distinguishing between
the two channels from their outputs is
\begin{equation}
	\frac{1+\epsilon/2}{2}.
\end{equation}

Next we present the single- and two-qudit gate set for this circuit construction.
Three types of single-qudit gates are specified by
\begin{equation}
	X_{jk} :=|j\rangle\langle k|+|k\rangle\langle j|,
\end{equation}
by the Givens rotation, which is a two-level unitary gate~\cite{NC00}
\begin{equation}
	G_{jk}(\theta) :=\cos\theta(|j\rangle\langle j|+|k\rangle\langle k|)
		+\sin\theta(|k\rangle\langle j|-|j\rangle\langle k|),
\end{equation}
and by the gate~$X_i$ from the Heisenberg-Weyl basis ($i,j,k\in \mathbb{Z}_d$).
Our gate notation implies an identity operator acting on the rest of the space.

We augment these gates by their two-qudit controlled counterparts
\begin{equation}
	CX_{jk}:=|j\rangle_\text{s}\langle j|\otimes X_{jk}
\end{equation}
and
\begin{equation}
	CG_{jk}(\theta):=|j\rangle_\text{s}\langle j|\otimes G_{jk}(\theta)
\end{equation}
with the system as control,
and
\begin{equation}
	CX_i:= X_i\otimes|i\rangle_\text{a}\langle i|
\end{equation}
with the ancilla as control.
We introduce a qudit multiplexer,
which generalizes the qubit case~\cite{SBM06},
as a sequence of two controlled Givens rotations
\begin{equation}
	M_{jk}(\alpha,\beta):=CG_{jk}(\alpha)CG_{kj}(-\beta)
\end{equation}
depicted in Fig.~\ref{fig:mult},
with the proof of the circuit equivalence
following straightforwardly
for the qubit case~\cite{BBC+95}.

\begin{figure}
\begin{indented}
\item[]
\includegraphics[width=12.5cm]{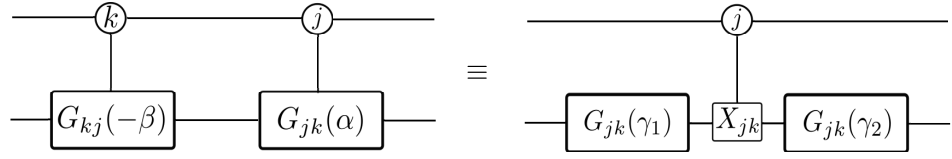}
\end{indented}
\caption{
	Circuit diagram for a multiplexer~$M_{jk}(\alpha,\beta)$.
	Each wire represents an evolving qudit in the~$d$-ary representation
	with~\textcircled{j} and~\textcircled{k}~$d$-ary control operation,
	and
	$\gamma_1\equiv\frac{1}{2}\left(\beta-\alpha+\frac{\pi}{2}\right)$
	and $\gamma_2\equiv\frac{1}{2}\left(\beta+\alpha-\frac{\pi}{2}\right)$.
}
\label{fig:mult}
\end{figure}

\begin{proposition}
\label{prop:cir}
Given any $\mathcal{V}(\bullet):=V\bullet V^\dagger$
and $\mathcal{W}(\bullet):=W\bullet W^\dagger$ with~$V,W\in SU(d)$,
any channel $\mathcal{W}\left(\text{tr}_\text{a}\mathcal{U}\right)\mathcal{V}$
is extreme provided that
\begin{equation}\label{eq:extU}
	U:= \prod_{i=d-1}^1 CX_i \prod_{j=d-1}^1\prod_{k=j-1}^0 M_{jk}(\alpha_{jk},\beta_{jk}),
\end{equation}
for all but a zero-measure subset of the rotation-angle sets~$\{\alpha_{jk}\}$ and~$\{\beta_{jk}\}$
with at most $(d^2-d)/2$ elements per set.
\end{proposition}
\begin{proof}
We prove the theorem by showing that the partial trace of~$U$ (\ref{eq:extU})
yields Kraus operators $F_i=_\text{a}\!\!\!\langle i| U|0\rangle_\text{a}$
that satisfy the normalization and linear independence conditions
of Prop.~\ref{prop:ext}.
To this end we define~$U'$ as a product of controlled-Givens rotations such that
\begin{equation}
	U=\left(\prod_{i=d-1}^1 CX_i\right)U'.
\end{equation}

We define
\begin{equation}
	\{u_{i\ell}\in \mathbb{R};i,\ell\in\mathbb{Z}_d\}
\end{equation}
such that
\begin{equation}
	U' |0\rangle_\text{a}|\ell\rangle_\text{s}
		= \sum_{i=0}^{d-1}u_{i\ell}|i\rangle_\text{a}|\ell\rangle_\text{s}.
\end{equation}
The unitary operator $U'$ corresponds to a channel
with diagonal Kraus operators $\{E_i\}$
such that
\begin{equation}
	E_i|\ell\rangle_\text{s}=u_{i\ell}|\ell\rangle_\text{s}
\end{equation}
as
\begin{align}
\label{eq:U'}
    U' |0\rangle_\text{a} |\ell\rangle_\text{s}
    = \sum_{i=0}^{d-1} |i\rangle_\text{a} \langle i| U' |0\rangle_\text{a} |\ell\rangle_\text{s}
    = \sum_{i=0}^{d-1} |i\rangle_\text{a} E_i |\ell\rangle_\text{s}.
\end{align}
We can identify $E_i$ in Eq.~(\ref{eq:U'})
with~$E_i$ in Eq.~(\ref{eq:F})
by setting
\begin{equation}
	u_{i\ell}\equiv\tilde{a}_{i\ell}:=\sum_{j=0}^{d-1}a_{ij}\operatorname{e}^{\operatorname{i}2\pi\ell j/d}.
\end{equation}

Reincorporating the gates $CX_i$ yields
\begin{align}
\label{eq:}
    U |0\rangle_\text{a} |\ell\rangle_\text{s}
    = \sum_{i=0}^{d-1} |i\rangle_\text{a} X_i E_i |\ell\rangle_\text{s}
    = \sum_{i=0}^{d-1} |i\rangle_\text{a} F_i |\ell\rangle_\text{s}.
\end{align}
A projection $|i\rangle_\text{a}\langle i|$ on the ancilla corresponds to the action of $F_i$ on the system.
The angles $\alpha_{jk}$ and $\beta_{jk}$ can be chosen (e.g., randomly) to satisfy
the linear independence of $\{F_i^\dagger F_j\}$.
This means the circuit $U$ realizes the Kraus operators $\{F_i\}$ for an extreme channel.
As there are ${d \choose 2}$ multiplexers,
the total number of independent parameters
is consistent with Corollary~\ref{cor:parameters}.
\end{proof}

As for constructing Kraus operators,
when the set $\{F_i^\dagger F_j\}$ is not linearly independent,
the circuit~(\ref{eq:extU}) realizes quasi-extreme channels.
As a result, the circuit in Prop.~\ref{prop:cir}
successfully yields simulations of arbitrary generalized extreme channels.

\section{Quantum channel simulation algorithm}
\label{sec:algo}

In this section, we first describe the classical optimization algorithm to design the quantum circuit,
and then we analyze the resultant quantum circuit in terms of space and time cost.

\subsection{Classical optimization algorithm to design the quantum circuit}

The first step in developing the quantum-circuit design procedure is to show the existence of
decomposition of a channel into a convex sum of generalized extreme channels.
Following from Ruskai's Conjectures~2, 3, 4 and 5~\cite{Rus07},
any channel $\mathcal{E}\in\mathscr{S}_d$ can be expressed as
\begin{equation}\label{eq:chadecom}
	\mathcal{E}=\sum_{\imath=1}^{d} p_\imath \mathcal{E}^\text{g}_\imath, \quad
\sum_{\imath=1}^{d} p_\imath=1, \quad 0 \leq p_\imath \leq 1,
\end{equation}
for generalized extreme channels $\{\mathcal{E}^\text{g}_\imath\}\subset\mathscr{S}_d$ of Kraus rank $\leq d$.
The upper bound for the convex sum is not guaranteed to be~$d$,
but this upper bound is implied by Ruskai's conjectures,
which we adopt in our algorithm.
Even assuming that the upper bound holds,
an analytical formula for such a decomposition is unknown.

Now we describe the algorithm for the simulation of a general qudit channel.
The algorithm accepts the dimension~$d$ of the Hilbert space,
the description of a channel $\mathcal{E}$
and an error tolerance $\epsilon$ as input.
The output is a quantum circuit
(with an output of zero reserved for the case that the algorithm aborts before finding a satisfactory circuit)
and a bound~$\tilde{\epsilon}$ on the resultant circuit with
respect to the actual channel~$\mathcal E$ being simulated.

Our algorithmic procedure is as follows.
Based on Ruskai's conjecture,
we assume that any given channel can be decomposed into
a~$d$-fold sum of generalized extreme channels~(\ref{eq:chadecom}),
and we know from Prop.~\ref{prop:cir}
a description of the circuit for any generalized extreme channel.
If the assumed decomposition does not hold, the algorithm can fail and an output of zero for the circuit description ensues.
Thus, Eq.~(\ref{eq:chadecom}) and Prop.~\ref{prop:cir} together
inform us that a quantum circuit for the qudit channel can be realized by
choosing generalized extreme channel circuits randomly with each $\imath^\text{th}$ circuit chosen
with probability~$p_\imath$.

Our algorithm initially chooses a set of~$d$ generalized extreme channels randomly
and tests whether the resultant guessed channel~$\tilde{\mathcal{E}}$ is within distance~$\epsilon$
of the correct channel~$\mathcal E$.
Typically the guessed channel fails to be within the error tolerance
so we employ an optimization algorithm to pick a new~$\tilde{\mathcal{E}}$
and try again.
This procedure is repeated until a satisfactory circuit is found or
aborted if the optimization routine fails to find a good circuit within a
pre-specified number of trials.

We now determine the number of parameters in~$\mathcal E$ for optimization.
From Eq.~(\ref{eq:chadecom})
we see that there are~$d-1$ parameters of~$\{p_\imath\}$.
The unitary matrices~$V$ and~$W$ in Prop.~\ref{prop:ext} could be constructed as
products
\begin{equation}
	V=\prod_i V_i,\;
	W=\prod_j W_j
\end{equation}
with as many unitary operators~$\varkappa$
in the two products as needed to provide enough parameters for the optimization.

As there are~$d$ generalized extreme channels and $d^2-1$ free parameters in $SU(d)$,
we have
\begin{equation}
	\varkappa d(d^2-1)
\end{equation}
free parameters associated with~$V$ and~$W$.
We add this number of parameters to the number of parameters for~$d$ generalized extreme channels,
namely $d(d^2-d)$ with~$d^2-d$ the number of free angles $\{\alpha_{jk},\beta_{jk}\}$,
and then add these to the number of probabilities~$\{p_\imath\}$.
The total number of parameters for the approximate channel should satisfy the inequality
\begin{equation}
\label{eq:count}
	\varkappa d(d^2-1)+d(d^2-d)+(d-1)\geq d^4-d^2
\end{equation}
with the right-hand side corresponding to the number of parameters that specify the qudit channel.
For the most efficient simulation,
we minimize~$\varkappa$ so
\begin{equation}
\label{eq:paracount}
	\varkappa=\left\lceil\frac{(d-1)(d^2+d+1)}{d(d+1)}\right\rceil.
\end{equation}
As an example, a qutrit channel has 72 parameters,
but our optimization is over 92 parameters.
Our analysis reduces to the qubit case~\cite{WBOS13}.
In that case~$d=2$ so the channel~$\mathcal E$ has 12 parameters
whereas the optimization of~$\mathcal E$ is over 17 parameters.
We note that the optimization precludes an efficient circuit-design algorithm even
in the qubit case, contrary to the earlier claim~\cite{WBOS13}.

The final step for the algorithm is to construct the objective function
for the optimization problem.
Mathematically we represent the correct $\mathcal{E}$ channel
by the Choi-Jamio{\l}kowski state~$\mathcal{C}$,
and the approximate circuit is represented by the state
\begin{equation}
\label{eq:C'}
	\mathcal{C}'=\sum_\imath p_\imath \mathcal{C}^\text{g}_\imath.
\end{equation}
Channel decomposition in the Choi-Jamio{\l}kowski state representation is elaborated
in~\ref{sec:choirep}.
Our goal is to find the best possible~$\mathcal{C}'$ by optimization.

The objective function for optimization is given by the trace distance~$D_t(\mathcal{C}, \mathcal{C}')$,
which bounds the $\diamond$-norm distance~(\ref{eq:diamond})
between two channels~$\mathcal{E}$
and~$\mathcal{E}'$ according to~\cite{Wat13}
\begin{equation}\label{eq:cbound}
	D_t(\mathcal{C}, \mathcal{C}')\leq \frac{\epsilon}{2}
\implies \|\mathcal{E}-\mathcal{E}'\|_{\diamond}\leq\epsilon.
\end{equation}
The trace distance is a convex function over the set of quantum states~\cite{NC00}.
Each $\mathcal{C}_\imath^\text{g}$ can be parameterized by a set of rotation angles
$\{\theta_{\imath\jmath}: \jmath=1,\dots, \varkappa (d^2-1)\}$ for the prior and posterior unitary operators,
and a set of rotation angles $\{\varphi_{\imath\jmath}: \jmath=1,\dots, d^2-d\}$,
which denote the sets $\{\alpha_{jk}\}$ and $\{\beta_{jk}\}$ altogether
from Eq.~(\ref{eq:extU}).
The range of the objective function is
\begin{equation}\label{eq:range}
	0\leq D_t(\mathcal{C}, \mathcal{C}')\leq d.
\end{equation}
The optimization is to find~$\mathcal{C}'$ such that $D_t$ is minimized according to
\begin{equation}
	\min_{\{\{p_\imath\},\{\theta_{\imath\jmath}\},\{\varphi_{\imath\jmath}\}\}}
		D_t(\mathcal{C}, \mathcal{C}').
\end{equation}
The minimization over probabilities~$\{p_\imath\}$~(\ref{eq:C'})
is subject to the constraint
$\sum_\imath p_\imath -1=0$.
Our algorithm employs a simple nonlinear programming method~\cite{YS06}
on channels generated by partial trace of Haar-random-generated unitary operators on the dilated
space~\cite{Toth08}.
We simulate on MATLAB$^\circledR$ using MultiStart and GlobalSearch algorithms;
simulated annealing was less effective.

We have demonstrated numerically that our optimization algorithm
is successful for systems of up to four dimensions.
Our simulations yield errors of order~$10^{-2}\sim10^{-4}$ for qubit channels,
$10^{-2}$ for qutrit channels,
and $10^{-1}$ for two-qubit channels.
The errors for the case $d=4$ from the numerical simulation is rather large yet acceptable for demonstrating
the efficacy of our algorithm.
For high-accuracy simulation,
significantly greater computational resources are required. 	
As the system dimension~$d$ increases,
we expect at least a quadratic increase in run-time of the simulation with respect to~$d$
due to the built-in method employed by GlobalSearch or MultiStart program.
Moreover, given that resources are finite, e.g.\ run-time,
numerical optimization is not even guaranteed to succeed due to becoming stuck at
certain points in the parameter space.
Such problems are quite generic for optimization problems.
In order to illustrate how the simulation works,
a concrete example for simulating one randomly chosen qutrit channel is presented in~\ref{sec:qutrit},
and the pseudo-code of the algorithm is presented in~\ref{sec:pseudocode}.

\subsection{Space and time cost of quantum simulation circuit}

Here we consider the time and space cost for the quantum circuit
to simulate a generalized extreme qudit channel~$\mathcal{E}^\text{g}$
on a quantum computer based on qudits and single- and two-qudit unitary gates.
The generalized extreme qudit channel is dilated to a unitary operator~$U$ on two qudits,
which contains a sequence of multiplexers and a sequence of $CX_i$ gates acting between the system
and the ancillary qudits,
and also a prior qudit rotation and a posterior qudit rotation acting on the system.

An arbitrary single qudit rotation can be decomposed into a product of at most $d(d-1)/2$,
which is $O(d^2)$, two-level unitary gates ($\S$4.5.1, \cite{NC00}).
A controlled-Givens rotation $CG_{jk}(\theta)$ can be realized by
two Givens rotations and a $CX_{jk}$ gate, similar to the qubit case~\cite{BBC+95}.
The sequence of $CX_i$ gates can be realized by classically controlled $X_i$ gates since the ancillary system is traced out,
with an $X_i$ gate acting on the system conditioned on a $|i\rangle\langle i|$ projection on the ancilla.
As a result, the generalized extreme channel circuit can be
realized by a product of $O(d^2)$ $CX_{jk}$ gates and
continuously-parameterized Givens rotations.

To assess the cost of the quantum circuit,
we employ the Solovay-Kitaev-Dawson-Nielsen algorithm for qudits~\cite{DN06}.
From the error tolerance $\epsilon$, which is an algorithmic input for circuit design,
any Givens rotation can be approximated by an
\begin{equation}
	O\left(\log\frac{d^{2}}{\epsilon}\right)
\end{equation}
sequence of universal qudit gates~\cite{DN06}.
As a result, the number of elementary gates,
hence computational time cost, of the generalized extreme channel circuit is
\begin{equation}
\label{eq:cost}
	O\left(d^2\log\frac{d^{2}}{\epsilon}\right),
\end{equation}
and the space cost is two qudits.

The circuit corresponding to~$\tilde{U}$
yields an approximation~$\mathcal{\tilde{E}}^\text{g}$
to the desired generalized extreme channel~$\mathcal{E}^\text{g}$.
From~\cite{WBOS13}
\begin{equation}\label{eq:wbos13}
\|\mathcal{E}^\text{g}-\mathcal{\tilde{E}}^\text{g}\|_{1\rightarrow 1}\leq2 \|U-\tilde{U}\|
=2\|U\otimes\mathds{1}-\tilde{U}\otimes\mathds{1}\|,
\end{equation}
we obtain
\begin{equation}
\label{eq:skbound}
	\|U-\tilde{U}\|\leq\frac{\epsilon}{2}\implies
		\|\mathcal{E}^\text{g}-\mathcal{\tilde{E}}^\text{g}\|_{\diamond}\leq\epsilon.
\end{equation}
From strong convexity and the chain property of trace distance,
relations (\ref{eq:cbound}) and~(\ref{eq:skbound}) above
together ensure the desired simulation accuracy~(\ref{eq:diamond}).

Finally simulating an arbitrary channel is implemented by
probabilistically implementing different generalized extreme channels
according to the distribution $\{p_\imath\}$~(\ref{eq:chadecom}).
The space and time costs of a single-shot implementation of the channel
are one dit and two qudits for space and
the classical time cost for generating the random dits plus
\begin{equation}
	O\left(d^2\log\frac{d^{2}}{\epsilon}\right)
\end{equation}
quantum gates.
In other words,
the quantum computational cost for simulating a random qudit channel is the same
as for simulating the generalized extreme channel,
and the additional cost is only classical:
dits plus running a random-number generator.

This cost can be explained by recognizing that the qudit channel simulator is
simply a randomized generalized extreme channel simulator.
On the other hand, estimating qudit observables accurately could require many shots,
with the number of shots depending on the particular observable.

\section{Conclusions}
\label{sec:conc}

In this work, we have presented a classical circuit-design algorithm for constructing
quantum simulation circuits for accurately simulating arbitrary qudit channels.
Our algorithm employs channel decomposition into a convex sum of generalized extreme channels,
leading to quantum circuits only for simulating generalized extreme channels,
which consume less computational resources.
In particular, we propose an ansatz for any extreme and quasi-extreme channels,
which has a concise mathematical structure.
We also show that the classical circuit-design algorithm can be
formalized as an optimization problem,
and we have performed numerical proof-of-principle simulations for low dimensional systems.

Our quantum channel simulation scheme transcends the standard quantum circuit model,
in that our circuit exploits resources other than quantum gates and measurements.
Using classical resources, e.g., bits, we can reduce the demand for quantum resources.
In particular, we show that it is possible to achieve the circuit lower bound~$O(d^2)$
for simulating non-unitary processes, which is the same lower bound for simulating
unitary operators.
Due to such a significant reduction of circuit cost,
our method is especially suitable for experimental implementation in the near future.

Our algorithm is designed for implementation in a fault-tolerant quantum computer
in order to achieve optimal efficiency,
but our simulation scheme can be implemented sooner with current technologies that are not fault-tolerant,
such as superconducting circuits, trapped ions or photons~\cite{BN09}.
Such systems admit more than two levels hence support qudits.
A proper error analysis for any physical system is required to assess its feasibility in the absence
of our fault-tolerance assumption.

This work generalizes the previous result of single-qubit-channel quantum simulation~\cite{WBOS13}
based on the extreme channel theory developed in Ref.~\cite{RSW02}.
Although a closed form of channel decomposition remains elusive,
our method provides an alternative approach to tackle this open problem
of channel decomposition into a convex sum of generalized extreme channels for the qudit case~\cite{Rus07}.
Although our algorithm relies on a conjectured upper bound to the number of
generalized extreme channels required to decompose the given channel~\cite{Rus07},
our numerical simulations succeed in delivering circuit designs that we verify work up to qudits of dimension 4.
If the conjectured upper bound for channel decomposition does not hold,
our algorithm can be modified by increasing the upper bound as an input to the algorithm,
but our algorithm is not guaranteed to be tractable under failure of this conjecture.
However, our algorithm is valuable as it works well for low dimensions,
and, if it fails at higher dimensions, superior algorithms could be tried.

\section{Acknowledgments}

	The authors acknowledges AITF, NSERC and USARO for financial support.
	This project was supported in part by the National Science Foundation under Grant No.\ NSF PHY11-25915.
    We thank D. W. Berry, R. Colbeck,
    M. C. de Oliveira, I. Dhand, R. Iten, R. Sweke, and E.\ Zahedinejad for
    both valuable discussions and critical comments.
    Also we thank J. Eisert for pointing out that the classical optimization is not convex.

\appendix

\section{Channel decomposition in Choi-Jamio{\l}kowski state representation}
\label{sec:choirep}

Here we present the channel decomposition~(\ref{eq:chadecom})
in Choi-Jamio{\l}kowski state representation~(\ref{eq:C}).
It turns out there exist nice block-matrix structures of Choi-Jamio{\l}kowski states
for general channels and extreme channels.
Then we present the channel decomposition with the block-matrix structures.

\begin{proposition}
\label{prop:block}
The Choi-Jamio{\l}kowski state~$\mathcal{C}$ for any qudit channel
$\mathcal{E}$ can be written in the block form
\begin{align}\label{eq:choiblock}
\mathcal{C}&=\begin{pmatrix}
           C_1 & C_{12} & \dots & C_{1d} \\
          \cdot & C_2 & \cdots & \cdot \\
          \vdots & \vdots & \ddots & \vdots \\
          \cdot & \cdot & \dots & C_{d} \\
         \end{pmatrix},
\end{align}
with~$d\times d$ positive matrices $C_k\geq0$ $(k=1,2,\dots, d)$
and
\begin{equation}\label{}
C_{kl}=C_{lk}^\dagger=\sqrt{C_k}B_{kl}\sqrt{C_l},
\end{equation}
with contraction $B_{kl}\leq 1$.
\end{proposition}
\begin{proof}
The block form follows from the following two facts:
i)~The positive semidefiniteness of a matrix is equivalent to the condition that all the principle minors are nonnegative.
Note that the principle minor is the determinant of the submatrices formed by columns and rows in the same set.
ii)~For the $2\times 2$ case~\cite{HJ91}, a block matrix
$$M=\begin{pmatrix} C_1 & C_{12} \\ C_{12}^\dagger & C_2 \\  \end{pmatrix} $$
is positive semidefinite iff $C_1,C_2\geq 0$, and $C_{12}=\sqrt{C_1}B_{12}\sqrt{C_2}$, for contraction $B_{12}\leq 1$.
The form~(\ref{eq:choiblock}) is a generalization of this case.
\end{proof}

Note that not all states in the form~(\ref{eq:choiblock}) are Choi-Jamio{\l}kowski states,
since Choi-Jamio{\l}kowski states are a particular kind of bipartite states.
The condition on the Choi-Jamio{\l}kowski state is that $\text{tr}_1 \mathcal{C}= \mathcal{E}(\mathds{1})$
and $\text{tr}_2 \mathcal{C}=\mathds{1}$,
where $\text{tr}_{1(2)} $ means the partial trace over the
1$\textsuperscript{st}$ (2$\textsuperscript{nd}$) part of the Choi-Jamio{\l}kowski state.
Furthermore, we also find the block-matrix form of generalized extreme Choi-Jamio{\l}kowski states,
in which the contraction matrices are substituted by unitary matrices.
\begin{proposition}
\label{Lemma:extrchoi}
The Choi-Jamio{\l}kowski state $\mathcal{C}^\emph{g}$
for any generalized extreme qudit channel $\mathcal{E}^\emph{g}$ can be written in the block form
\begin{align}\label{eq:extchoiblock}
\mathcal{C}^\emph{g}&=\begin{pmatrix}
           C_1^\emph{g} & C_{12}^\emph{g} & \dots & C_{1d}^\emph{g} \\
          \cdot & C_2^\emph{g} & \dots & \cdot \\
          \vdots & \vdots & \ddots & \vdots \\
          \cdot & \cdot & \dots & C_{d}^\emph{g} \\
         \end{pmatrix},
\end{align}
with~$d\times d$ positive matrices $C^\emph{g}_k\geq0$ $(k=1,2,\dots, d)$,
\begin{equation}\label{}
C_{kl}^\emph{g}=\sqrt{C_k^\emph{g}}U_{kl}\sqrt{C^\emph{g}_l},
\end{equation}
and unitary operators
\begin{equation}\label{}
U_{kl}:=\prod^{l-1}_{s=k}U_{s,s+1}
\end{equation}
with unitary operators $U_{s,s+1}\in SU(d)$.
\end{proposition}
\begin{proof}
The matrix $\mathcal{C}^\text{g}$ can be decomposed as $\mathcal{C}^\text{g}=AUA$, with
\begin{align}
A=\begin{pmatrix}
           \sqrt{C^\text{g}_1} &  &  & \makebox(0,0){\text{\huge0}}\\
           & \sqrt{C^\text{g}_2} &  &  \\
           &   & \ddots &  \\
          \text{\huge0}  &  & & \sqrt{C^\text{g}_{d}} \\
         \end{pmatrix},
U=\begin{pmatrix}
           \mathds{1} & U_{12} & \dots & U_{1d} \\
          \cdot & \mathds{1} & \dots & \cdot \\
          \vdots & \vdots & \ddots & \vdots \\
          \cdot & \cdot & \dots & \mathds{1} \\
         \end{pmatrix},
\end{align}
where $A$ a diagonal block matrix, and $U$ can be further written as
$U=\tilde{U}^\dagger \tilde{\mathds{1}}\tilde{U}$, with
\begin{align}
\tilde{U}=\begin{pmatrix}
           \mathds{1} & U_{12} & \dots & U_{1d} \\
           & \mathds{1} & \dots & \cdot \\
           &  & \ddots & \vdots \\
           & \text{\huge0}  &  & \mathds{1} \\
         \end{pmatrix},
\tilde{\mathds{1}}=\begin{pmatrix}
           \mathds{1} &  &  & \makebox(0,0){\text{\huge0}}  \\
           & \bm{0} &  &  \\
           &   & \ddots &  \\
           \text{\huge0} &   & & \bm{0} \\
         \end{pmatrix},
\end{align}
where $\tilde{U}$ is an upper triangular matrix,
$\bm{0}$ represents zero block matrix,
and the large zeros represents zero entries.
As a result,
\begin{equation}
	\mathcal{C}^\text{g}=A\tilde{U}^\dagger \tilde{\mathds{1}}\tilde{U}A.
\end{equation}
This factorization implies that $\mathcal{C}^\text{g}$ is positive semidifinite,
and its rank is bounded above by~$d$,
which is the same as the number of Kraus operators for the generalized extreme channel.
Note that $\text{tr}_1 \mathcal{C}^\text{g}= \mathcal{E}^\text{g}(\mathds{1})$
and $\text{tr}_2 \mathcal{C}^\text{g}=\mathds{1}$.
\end{proof}

\begin{corollary}
\label{cor:choidecom}
The Choi-Jamio{\l}kowski state~$\mathcal{C}$ for any qudit channel $\mathcal{E}\in\mathscr{S}_d$
can be expressed as
\begin{equation}
\label{eq:choidecom}
	\mathcal{C}=\sum_{\imath=1}^{d} p_\imath \mathcal{C}^\emph{g}_\imath, \quad
\sum_{\imath=1}^{d} p_\imath=1, \quad 0\leq p_\imath \leq 1,
\end{equation}
with each $\mathcal{C}^\emph{g}_\imath$ corresponding to
a generalized extreme channel $\mathcal{E}^\emph{g}_\imath$.
\end{corollary}
\begin{proof}
Following from Eq.~(\ref{eq:chadecom})
based on Ruskai's conjectures~\cite{Rus07} we have
\begin{equation}
	\mathcal{C}=\sum_{\imath=1}^{d} p_\imath \mathcal{C}^\text{g}_\imath
\end{equation}
with each generalized extreme state $\mathcal{C}^\text{g}_\imath$ corresponding to
a generalized extreme channel $\mathcal{E}^\text{g}_\imath$.
With block-matrix forms~(\ref{eq:choiblock}) and (\ref{eq:extchoiblock}),
Eq.~(\ref{eq:choidecom}) is equivalent to
\begin{subequations}
\begin{align}\label{}
  C_k & = \sum_{\imath=1}^{d} p_\imath C^\text{g}_{k;\imath}, \\
  C_{kl} & = \sum_{\imath=1}^{d} p_\imath C^\text{g}_{kl;\imath},
\end{align}
\end{subequations}
for $k, l=1,2,\dots,d$, $k<l$.
\end{proof}

\section{Simulation example for qutrit channels}
\label{sec:qutrit}

For qutrit channels,
our classical algorithm accepts the description of any qutrit channel and an error tolerance~$\epsilon$,
quantified by diamond norm~\cite{KSV02},
and proceeds by decomposing it into a convex combination of generalized extreme channels,
and then delivers quantum circuits for its simulation.
The quantum circuit processes any qutrit system state
to generate output state within given error tolerance for the worst case.

The input qutrit channel $\mathcal{E}$ needs to be chosen randomly,
and this is realized by randomly generating unitary operator in $SU(27)$ according to Haar measure,
since a channel can be realized by unitary operator followed by a partial trace on the environment~\cite{Sti55}.
We employ the MATLAB$^\circledR$ package~\cite{Toth08}.
For instance, a unitary operator $U\in SU(27)$ is generated with command ``runitary(3,3)'' in MATLAB$^\circledR$.

The set of Kraus operators for the input channel is obtained from the relation $K_i=\langle i|U|0\rangle$.
Then, the Choi-Jamio{\l}kowski state~$\mathcal{C}$ is obtained from Eq.~(\ref{eq:C}).
Also, we can use reshaping operation~\cite{BZ06,Mis11},
which is defined as, given an $m\times m$ matrix $A=[a_{ij}]$ with elements $a_{ij}$,
\begin{equation}\label{}
\text{res}(A)=(a_{11},\dots,a_{1m},\dots,a_{m1},\dots,a_{mm})^\text{T}.
\end{equation}
The Choi-Jamio{\l}kowski state takes the form
\begin{equation}\label{}
\mathcal{C}=\sum_i \text{res}(K_i)\cdot \text{res}(K_i)^\dagger.
\end{equation}

For example, we generate a unitary operator (not shown here since it is $27 \times 27$),
and obtain the Choi-Jamio{\l}kowski state for an input channel as
\begin{scriptsize}
\begin{align}\label{eq:3Cexample}\nonumber
  \mathcal{C}= &
\left(\begin{matrix}
   0.3105 + 0.0000i &  0.1052 + 0.0154i&   0.0394 - 0.0099i&   0.0554 + 0.0013i&  -0.0892 - 0.0667i\\
0.1052 - 0.0154i&   0.2526 + 0.0000i & -0.0174 + 0.0148i &  0.0715 - 0.0388i & -0.0307 - 0.0814i\\
0.0394 + 0.0099i & -0.0174 - 0.0148i  & 0.2473 + 0.0000i  & 0.0302 + 0.0637i  & 0.0600 + 0.0425i\\
0.0554 - 0.0013i  & 0.0715 + 0.0388i   &0.0302 - 0.0637i &  0.2891 + 0.0000i&  -0.0066 - 0.0301i\\
-0.0892 + 0.0667i  &-0.0307 + 0.0814i&   0.0600 - 0.0425i & -0.0066 + 0.0301i&   0.2667 + 0.0000i \\
0.0185 - 0.0149i&  -0.1021 - 0.0296i  & 0.0800 - 0.1476i&  -0.0147 - 0.0141i  & 0.0816 - 0.0849i\\
-0.0070 - 0.0066i&  -0.0250 - 0.0841i  &-0.0986 - 0.0101i&  -0.0501 + 0.0559i  &-0.0968 - 0.1558i\\
-0.1068 - 0.1226i &  0.0778 + 0.0614i&   0.0475 + 0.0007i &  0.1728 + 0.0571i&   0.0234 + 0.0100i\\
0.1131 - 0.0192i & -0.0717 - 0.0218i & -0.0440 - 0.0073i   &0.1132 - 0.0481i  &-0.0568 + 0.0703i
\end{matrix}\right. \\
   & \left.\begin{matrix}
     0.0185 + 0.0149i & -0.0070 + 0.0066i& -0.1068 + 0.1226i&   0.1131 + 0.0192i \\
   -0.1021 + 0.0296i & -0.0250 + 0.0841i&  0.0778 - 0.0614i&  -0.0717 + 0.0218i\\
    0.0800 + 0.1476i  &-0.0986 + 0.0101i & 0.0475 - 0.0007i & -0.0440 + 0.0073i\\
   -0.0147 + 0.0141i&  -0.0501 - 0.0559i  &0.1728 - 0.0571i  & 0.1132 + 0.0481i\\
   0.0816 + 0.0849i  &-0.0968 + 0.1558i&  0.0234 - 0.0100i & -0.0568 - 0.0703i\\
    0.3716 + 0.0000i  & 0.0306 + 0.1539i&  -0.1073 - 0.0026i&   0.0882 + 0.0653i\\
    0.0306 - 0.1539i   &0.4004 + 0.0000i & -0.0986 + 0.0147i & -0.0247 - 0.0042i\\
   -0.1073 + 0.0026i&  -0.0986 - 0.0147i  &0.4807 + 0.0000i&  -0.0641 - 0.0998i\\
     0.0882 - 0.0653i&  -0.0247 + 0.0042i  &-0.0641 + 0.0998i &  0.3811 + 0.0000i
\end{matrix}\right).
\end{align}
\end{scriptsize}

\noindent The matrix~$\mathcal{C}$ is positive with eigenvalues
 0.0018,    0.0244,    0.0662,    0.1366,    0.2499,
 0.4415,    0.5808,    0.6519,    0.8469.

According to our decomposition method,
we need to approximate the Choi-Jamio{\l}kowski state~$\mathcal{C}$ by
\begin{equation}\label{eq:choidecom3}
\mathcal{C}'=\sum_{i=1}^{3} p_i \mathcal{C}^\text{g}_i, \quad
\sum_{i=1}^{3} p_i=1, \quad 0 \leq p_i \leq 1,
\end{equation}
and each $\mathcal{C}^\text{g}_i$ corresponds to one generalized extreme channel.
Each generalized extreme channel is specified by three Kraus operators
\begin{subequations}\label{eq:qutrit}
\begin{equation}
F_0=
\begin{pmatrix}
          \cos a \cos c & 0 & 0 \\
          0 & \cos b & 0 \\
          0 & 0 & \cos d \\
         \end{pmatrix},
\end{equation}
\begin{equation}\label{}
F_1=\begin{pmatrix}
          0 & \sin b\cos e & 0 \\
          0 & 0 & -\sin d \sin f \\
          \sin a & 0 & 0 \\
         \end{pmatrix},
\end{equation}
\begin{equation}\label{}
F_2=\begin{pmatrix}
          0 & 0 & \sin d \cos f \\
          \cos a \sin c & 0 & 0 \\
          0 & \sin b \sin e & 0 \\
         \end{pmatrix},
\end{equation}
\end{subequations}
with parameters $0\leq a,b,c,d,e,f\leq 2\pi$, and three initial and final unitary operators,
denoted as $R_1$, $R_2$, and $R_3 \in SU(3)$ acting on the system.
We let there be two initial unitary operators $R_3$ followed by~$R_2$,
and one final unitary operator $R_1$.
The unitary operator in $SU(3)$ which has 8 real parameters~\cite{NV12} is parameterized as
\begin{equation}
U=
\begin{pmatrix}
          e^{i\phi_1}c_1c_2 & e^{i\phi_3}s_1 & e^{i\phi_4}c_1s_2 \\
          e^{-i\phi_4-i\phi_5}s_2s_3-e^{i\phi_1+i\phi_2-i\phi_3}s_1
          c_2c_3 & e^{i\phi_2}c_1c_3 & -e^{-i\phi_1-i\phi_5}c_2s_3-e^{i\phi_2-i\phi_3+i\phi_4}s_1
          s_2c_3 \\
          -e^{-i\phi_2-i\phi_4}s_2c_3-e^{i\phi_1-i\phi_3+i\phi_5}s_1
          c_2s_3 & e^{i\phi_5}c_1s_3 & e^{-i\phi_1-i\phi_2}c_2c_3-e^{-i\phi_3+i\phi_4+i\phi_5}s_1
          s_2s_3
         \end{pmatrix},
\end{equation}
where $c_k\equiv\cos \theta_k$ and $s_k\equiv\sin \theta_k$.
There are three Euler angles $\theta_j$ ($0\leq \theta_j \leq \pi/2$; $j=1,2,3$)
and five phases $\phi_k$ ($0\leq \phi_k \leq 2\pi$; $k=1,2,3,4,5$).

For one qutrit generalized extreme channel, the circuit executes the operator
\begin{equation}\label{eq:ext3}
U_\text{qut}:=CX_2CX_1M_{21}(\alpha_{21},\beta_{21})M_{20}(\alpha_{20},\beta_{20})M_{10}(\alpha_{10},\beta_{10}),
\end{equation}
together with a prior and a posterior qutrit rotation acting on the system.
The three multiplexers contain six controlled-Givens rotations,
which can be implemented by six $CX_{jk}$ gates and nine Givens rotations
obtained from basic technique of circuit design~\cite{BBC+95}.
The prior and posterior qutrit rotations acting on the system
can be realized by six Givens rotations
since a qutrit rotation can be decomposed as a product of three Givens rotations each acting on a qubit subspace of the qutrit~\cite{NV12}.
The gates $X_1$ in $CX_1$ can be decomposed as $X_{21}X_{10}$,
and $X_2$ can be decomposed as $X_{21}X_{20}$,
and then
\begin{equation}\label{eq:3cx}
  CX_2CX_1=CX_{21}CX_{20}CX_{21}CX_{10}
\end{equation}
with the ancilla as control for all of the four $CX_{jk}$ gates.
In all, there are 10 $CX_{jk}$ gates and 15 Givens rotations in the circuit for a qutrit generalized extreme channel.
Furthermore, if classical feedback is available,
the last four $CX_{jk}$ gates can be replaced by classically controlled $X_{jk}$ gates,
as the case of qubit channel simulation~\cite{WBOS13}.

The optimization is implemented such that the trace distance
$D_t(\mathcal{C}, \mathcal{C}')\leq \epsilon/2$.
Algorithms such as MultiStart, GlobalSearch or Simulated Annealing in MATLAB$^\circledR$ are employed,
and we choose the best simulation results.

\begin{table}
\scriptsize
\addtolength{\tabcolsep}{-1.5pt}
\caption{\label{tab:qutrit}
  The simulation result for the decomposition of a randomly generated qutrit channel in Eq.~(\ref{eq:3Cexample}).
  The table on the left contains the parameters for prior and posterior unitary operators,
  and the table on the right contains the parameters ($a,b,c,d,e,f$) for Kraus operators and the probability $p$,
  and also the three eigenvalues $\lambda$ for each generalized extreme channels.
}
\noindent
\subtable{
\begin{tabular}{|c||c|c|c|c|c|c|c|c|c|}
  \hline
  & \multicolumn{3}{|c|}{$\mathcal{C}^\text{g}_1$}
  & \multicolumn{3}{|c|}{$\mathcal{C}^\text{g}_2$} & \multicolumn{3}{|c|}{$\mathcal{C}^\text{g}_3$} \\ \hline
  & $R_1$ & $R_2$ & $R_3$ & $R_1$ & $R_2$ & $R_3$ & $R_1$ & $R_2$ & $R_3$ \\ \hline  \hline
  $\theta_1$ & 1.2344 & 0.2292 & 0.9562 & 0.6197 & 0.3668 & 1.1377 & 0.3082 & 1.1345 & 0.3574 \\ \hline
  $\theta_2$ & 1.2781 & 0.6352 & 0.1978 & 1.1258 & 0.4069 & 1.1456 &0.6092  & 0.5117 & 1.1794 \\ \hline
  $\theta_3$ & 0.6618 & 0.4768 & 0.5194 & 0.9545 & 0.0651 & 1.4608 & 1.4406 & 0.6749 & 0.3582 \\ \hline
  $\phi_1$ & 1.1865 & 4.0185 & 2.6995 & 4.0777 & 1.8266 & 3.9186 &4.9068  & 4.7498 & 2.1275 \\ \hline
  $\phi_2$ & 4.1535 & 3.3050 & 5.1831 & 1.8561 & 4.9335 & 2.3296 &5.4269  & 3.1036 & 2.5366 \\ \hline
  $\phi_3$ & 1.6894 & 5.0089 & 2.1618 & 3.6516 & 2.8210 & 4.3385 &2.6902  & 5.3635 & 5.1105 \\ \hline
  $\phi_4$ & 0.8490 & 2.1711 & 3.9187 & 4.9058 & 2.1526 & 5.4539 & 2.8977 & 4.5586 & 3.3091 \\ \hline
  $\phi_5$ & 4.7523 & 3.6288 & 1.2381 & 2.2728 & 3.1790 & 3.1468 & 0.6585 & 3.6124 & 2.2142 \\
  \hline
\end{tabular}
}
\subtable{
\begin{tabular}{|c||c|c|c|}
  \hline
  & $\mathcal{C}^\text{g}_1$ & $\mathcal{C}^\text{g}_2$ & $\mathcal{C}^\text{g}_3$ \\ \hline \hline
  $a$ & 2.1417 & 2.3442 & 2.0610 \\ \hline
  $b$ & 4.8284 & 1.8620 & 3.9621 \\ \hline
  $c$ & 2.3434 & 4.7272 & 1.6220 \\ \hline
  $d$ & 4.0164 & 2.2822 & 1.0321 \\ \hline
  $e$ & 2.7418 & 4.0726 & 2.5719 \\ \hline
  $f$ & 3.1900 & 4.8792 & 5.2118 \\ \hline
  $\lambda_1$ & 0.5667& 0.5088 & 0.5457 \\ \hline
  $\lambda_2$ & 0.8868& 1.0942 & 0.7287 \\ \hline
  $\lambda_3$ & 1.5465 & 1.3970 & 1.7256 \\ \hline
  $p$ & 0.2974 & 0.3676 & 0.3350 \\
  \hline
\end{tabular}
}
\end{table}

Next we present one simulation result,
which contains 92 parameters:
\begin{itemize}
  \item 30 for each generalized extreme channel:
  6 in the Kraus operators and 24 in the initial and final unitary operators;
  \item 2 for the probability $p_1$, $p_2$.
\end{itemize}
The result is summarized in Table~\ref{tab:qutrit}.
The approximate Choi-Jamio{\l}kowski state $\mathcal{C}'$ is
\begin{scriptsize}
\begin{align}\label{}\nonumber
  \mathcal{C}'= &
\left(\begin{matrix}
   0.3103 + 0.0000i &  0.1082 + 0.0119i &  0.0386 - 0.0089i &  0.0559 + 0.0007i & -0.0859 - 0.0676i\\
   0.1082 - 0.0119i &  0.2522 + 0.0000i & -0.0243 + 0.0256i &  0.0726 - 0.0333i & -0.0393 - 0.0777i \\
   0.0386 + 0.0089i & -0.0243 - 0.0256i &  0.2520 + 0.0000i &  0.0309 + 0.0603i &  0.0645 + 0.0313i \\
   0.0559 - 0.0007i &  0.0726 + 0.0333i &  0.0309 - 0.0603i &  0.2951 + 0.0000i & -0.0095 - 0.0290i \\
  -0.0859 + 0.0676i & -0.0393 + 0.0777i &  0.0645 - 0.0313i & -0.0095 + 0.0290i &  0.2677 + 0.0000i \\
   0.0207 - 0.0164i & -0.0922 - 0.0272i &  0.0765 - 0.1407i & -0.0133 - 0.0100i &  0.0871 - 0.0753i\\
  -0.0090 - 0.0044i & -0.0260 - 0.0903i & -0.1034 - 0.0120i & -0.0521 + 0.0552i & -0.0975 - 0.1628i \\
  -0.1058 - 0.1225i &  0.0797 + 0.0632i &  0.0461 - 0.0006i &  0.1708 + 0.0550i &  0.0246 + 0.0135i\\
   0.1126 - 0.0218i & -0.0676 - 0.0221i & -0.0414 - 0.0107i &  0.1136 - 0.0481i & -0.0505 + 0.0714i
\end{matrix}\right. \\
   &  \left.\begin{matrix}
 0.0207 + 0.0164i& -0.0090 + 0.0044i & -0.1058 + 0.1225i &  0.1126 + 0.0218i\\
 -0.0922 + 0.0272i&  -0.0260 + 0.0903i &  0.0797 - 0.0632i & -0.0676 + 0.0221i\\
 0.0765 + 0.1407i&  -0.1034 + 0.0120i &  0.0461 + 0.0006i & -0.0414 + 0.0107i\\
 -0.0133 + 0.0100i & -0.0521 - 0.0552i &  0.1708 - 0.0550i &  0.1136 + 0.0481i\\
  0.0871 + 0.0753i &-0.0975 + 0.1628i &  0.0246 - 0.0135i & -0.0505 - 0.0714i\\
  0.3731 + 0.0000i &0.0329 + 0.1523i & -0.1037 - 0.0034i &  0.0828 + 0.0638i\\
  0.0329 - 0.1523i &0.3946 + 0.0000i & -0.0987 + 0.0171i & -0.0253 - 0.0012i\\
 -0.1037 + 0.0034i &-0.0987 - 0.0171i &  0.4802 + 0.0000i & -0.0628 - 0.1009i \\
  0.0828 - 0.0638i & -0.0253 + 0.0012i & -0.0628 + 0.1009i &  0.3749 + 0.0000i
\end{matrix}\right),
\end{align}
\end{scriptsize}

\noindent with eigenvalues
0.0039,    0.0280,    0.0797,    0.1264,    0.2473,
0.4395,    0.5825,    0.6515,  and  0.8413.
The actual error (the trace distance between input Choi-Jamio{\l}kowski state~$\mathcal{C}$
and approximate Choi-Jamio{\l}kowski state $\mathcal{C}'$) is 0.046.
This means the probability to distinguish the true channel from the approximate one
is $\frac{1}{2}(1+0.046)$.
As $0.046\ll 1$, this indicates that the channel decomposition is good enough for accurate simulation.
We have performed simulations for about~$50$ randomly chosen channels,
and the errors are all in the order~$0.01$.

In addition, we have checked the block-matrix structure to verify our simulations.
That is, in the block-matrix form a generalized extreme channel takes the form
\begin{equation}\label{eq:choiqutrit}
\mathcal{C}^\text{g}=
\begin{pmatrix}
C_1 & \sqrt{C_1}U_{12}\sqrt{C_2} & \sqrt{C_1}U_{13}\sqrt{C_3} \\
\sqrt{C_2}U_{12}^\dagger\sqrt{C_1} & C_2 & \sqrt{C_2}U_{23}\sqrt{C_3} \\
\sqrt{C_3}U_{13}^\dagger\sqrt{C_1} & \sqrt{C_3}U_{23}^\dagger\sqrt{C_2} & C_3
\end{pmatrix},
\end{equation}
for positive matrices $C_1$, $C_2$, $C_3$,
and unitary operators $U_{12}$, $U_{23}$,
and $U_{13}=U_{12}U_{23}$.

\section{The Algorithm}
\label{sec:pseudocode}

We first explain our notation prior to presenting the pseudo-code.
We use $[\bullet]$ to denote a bit-string description of some object.
For instance, $[\mathcal{E}]$ is the description for a channel $\mathcal{E}$,
and we use $[C]$ to denote a quantum circuit.

We use $\textsc{ChaSim}$ to denote the main algorithm for channel simulation,
$\textsc{CJ}$ to denote the Choi-Jamio{\l}kowski state optimization decomposition,
and $\textsc{SK}(U,\epsilon)$ to denote the Solovay-Kitaev algorithm for the approximation
of a gate $U$ containing continuous variables, e.g., rotation angles,
by gates from a universal library within spectrum-norm distance $\epsilon$.
We denote the multi-use qudit Solovay-Kitaev algorithm by
$\textsc{SK}(U_1, U_2, \cdots, U_m, \epsilon)$ for the gates $U_1, U_2, \cdots, U_m$, respectively,
so that each of the gate is approximated with error input $\epsilon$.

\begin{algorithm}
\caption{Algorithm for qudit quantum channel $\mathcal{E}$ simulation}
\label{alg:sqd}
\begin{footnotesize}
\begin{algorithmic}
\Require \\
$[\mathcal{E}]$: bit-string description of the channel\\
$d$: the qudit dimension\\
$\epsilon$: the error tolerance
\Ensure \\
$[C]$: bit-string description of the circuit
\Function {ChaSim}{$[\mathcal{E}]$,~$d$, $\epsilon$}
\State  $[C]\gets \varnothing$.  \Comment{Initializes as the empty-string}
\State $U \gets \text{Haar-rand-}SU(d^3)$. \Comment{Generate random unitary operator}
\For {$i=0$ to $d^2-1$}
\State $K_i \gets \langle i| U |0\rangle $. \Comment{Generate Kraus operators}
\EndFor
\State $ [\mathcal{E}] \gets  \{K_i\} $. \Comment{Generate input channel}
\State  $\mathcal{C} \gets [\mathcal{E}]$.
\Comment{Convert channel to Choi-Jamio{\l}kowski state~$\mathcal{C}$}
\For {$\imath=1$ to~$d$}
\State $\vec{p} \gets \text{rand} [0,1 ]^{\otimes d}$. \Comment{Generate probability}
\State $W^{(\imath)}, V^{(\imath)} \gets \text{Haar-rand-}SU(d)$.
\Comment{Generate random unitary operators}
\For {$i=1$ to $d-1$}
\For {$j=0$ to $i-1$}
\State $\vec{\theta}^{(\imath)} \gets 2\pi \text{rand} [0,1 ]^{\otimes d^2-d}$.
\Comment{Generate rotation angles}
\State $G_{ij}^{(\imath)} (\theta_{ij}^{(\imath)}) \gets \cos \theta_{ij}^{(\imath)}
(|i\rangle\langle i|+|j\rangle\langle j|) + \sin \theta_{ij}^{(\imath)} (|j\rangle\langle i|-|i\rangle\langle j|) $.
\State $CG_{ij}^{(\imath)} (\theta_{ij}^{(\imath)}) \gets |i\rangle\langle i| \otimes G_{ij}^{(\imath)} (\theta_{ij}^{(\imath)})  $.
\EndFor
\EndFor
\For {$i=1$ to $d-1$}
\State $X_i \gets \sum_{\ell=0}^{d-1} |\ell\rangle\langle \ell+i|$.
\State $CX_i \gets X_i\otimes |i\rangle\langle i|$. \Comment{Controlled-$X_i$ gates}
\EndFor
\State $U^{(\imath)} \gets \prod_{i=1}^{d-1} CX_i \prod_{i=d-1}^{1}
\prod_{j=i-1}^{0} CG_{ij}^{(\imath)} (\theta_{ij}^{(\imath)})  CG_{ji}^{(\imath)} (\theta_{ji}^{(\imath)}) $.
\For {$i=0$ to $d-1$}
\State $F_i^{(\imath)} \gets \langle i|U^{(\imath)}|0\rangle$.
\State $K_i^{(\imath)}  \gets W^{(\imath)}F_i^{(\imath)}  V^{(\imath)}$.
\Comment{Kraus operators for each generalized extreme channel}
\EndFor
\State $\mathcal{C}^{(\imath)} \gets \{K_i^{(\imath)}\}$.
\EndFor
\State $\{\epsilon', \vec{p}', \vec{\theta}^{(\imath)} , W^{(\imath)},  V^{(\imath)}\} \gets$
$\textsc{CJ}(\mathcal{C}, \epsilon, \vec{p}, \{\mathcal{C}^{(\imath)}\})$.
\Comment{Choi-Jamio{\l}kowski state decomposition }
\If {$\epsilon'\leq \epsilon$}
\State \Return $U^{(\imath)}  \gets \vec{\theta}^{(\imath)} , W^{(\imath)},  V^{(\imath)}$.
\For {$\imath=1$ to~$d$}
\State $\tilde{W}^{(\imath)}, \tilde{V}^{(\imath)}, \tilde{G}_{ij}^{(\imath)} \gets$
$\textsc{SK}(W^{(\imath)}, V^{(\imath)}, G_{ij}^{(\imath)}, \epsilon)$.
\Comment{Solovay-Kitaev algorithm}
\State $[C^{(\imath)}] \gets \tilde{W}^{(\imath)}, \tilde{V}^{(\imath)}, \tilde{G}_{ij}^{(\imath)}$.
\Comment{Construct the generalized extreme channel circuit}
\EndFor
\State \Return $[C]\gets [C^{(1)}][C^{(2)}]\cdots[C^{(d)}][\vec{p}']$.
\Else
\State \Return false.
\EndIf
\EndFunction
\end{algorithmic}
\end{footnotesize}
\end{algorithm}

\clearpage

\section*{References}
\bibliography{ext}

\end{document}